\documentclass[reqno]{amsart}
\usepackage{amsmath}
\usepackage{amsfonts}
\usepackage{amssymb}
\usepackage{amssymb,graphics,graphicx,bbm,hyperref,color}
\usepackage{multicol}
\usepackage{enumitem} 
\usepackage{dsfont}
\usepackage{mathrsfs}

\def\be{\begin{equation}}
\def\ee{\end{equation}}
\newcommand{\bbC}{{\mathbb C}}

\newcommand{\bbZ}{{\mathbb Z}}
\newcommand{\cl}{{\rm{cl}}}
\newcommand{\qu}{{\rm{qu}}}
\newcommand{\scrS}{\mathscr{S}}
\newcommand{\caH}{\mathcal{H}}

\def\Tr{{\operatorname{Tr\,}}}
\newcommand{\e}[1]{\,{\rm e}^{#1}\,}
\newtheorem{thm}{Theorem}
\newtheorem{cor}[thm]{Corollary}
\newtheorem{lem}[thm]{Lemma}
\numberwithin{thm}{section}

\begin{document}


\title{Correlation inequalities for classical and quantum XY models}

\author{Costanza Benassi}
\author{Benjamin Lees}
\author{Daniel Ueltschi}
\address{Department of Mathematics, University of Warwick,
Coventry, CV4 7AL, United Kingdom}
\email{C.Benassi@warwick.ac.uk}
\address{Department of Mathematics, University of Gothenburg, Sweden}
\email{benjaminlees90@gmail.com}
\address{Department of Mathematics, University of Warwick,
	Coventry, CV4 7AL, United Kingdom}
\email{daniel@ueltschi.org}

\maketitle

\begin{abstract}We review correlation inequalities of truncated functions for the classical and quantum XY models. A consequence is that the critical temperature of the XY model is necessarily smaller than that of the Ising model, in both the classical and quantum cases. We also discuss an explicit lower bound on the critical temperature of the quantum XY model.\end{abstract}

\section{Setting and results}

The goal of this survey is to recall some results of old that have been rather neglected in recent years. We restrict ourselves to the cases of classical and quantum XY models.
Correlation inequalities are an invaluable tool that allows to obtain the monotonicity of spontaneous magnetisation, the existence of infinite volume limits, and comparisons between the critical temperatures of various models. Many correlation inequalities have been established for the planar rotor (or classical XY) model, with interesting applications and consequences in the study of the phase diagram and the Gibbs states \cite{Dun,KPV75,Mon,MP,KPV76,MMP,FP}.
Some of these inequalities can also be proved for its quantum counterpart \cite{Gal,Suz,PM2,BLU}.

Let $\Lambda$ be a finite set of sites. The classical XY model (or planar rotor model) is a model for interacting spins on such a lattice. The configuration space of the system is defined as $\Omega_\Lambda = \{\{\sigma_x\}_{x\in\Lambda}: \sigma_x\in S^1 \,\,\forall x\in\Lambda\}$: each site hosts a unimodular vector lying on a unit circle. It is convenient to represent the spins by means of angles, namely
\begin{eqnarray}
\sigma_x^1 &=& \cos{\phi_x}\label{cos}\\
\sigma_x^2 &=& \sin{\phi_x}\label{sin}
\end{eqnarray}
with $\phi_x\in\left[0,2\pi\right]$.
The energy of a configuration $\sigma \in \Omega_\Lambda$ with angles $\phi =~\{\phi_x\}_{x\in\Lambda} $ is
\be
H_\Lambda^{\rm{cl}}(\phi) = -\sum_{A \subset \Lambda} J^1_A  \prod_{x \in A} \sigma^1_x + J^2_A \prod_{x \in A} \sigma^2_x,\label{classicalXY}
\ee
with $J^i_A\in\mathbb{R}$ for all $A\subset\Lambda$. The expectation value at inverse temperature $\beta$ of a functional $f$ on the configuration space is
\begin{equation}
\langle f\rangle^{\rm cl}_{\Lambda,\beta} = \frac{1}{Z^{\rm cl}_{\Lambda,\beta}}\int d\phi \,\mbox{e}^{-\beta H^{\rm{cl}}_\Lambda(\phi) }\,f(\phi),
\end{equation}
where $Z^{\rm cl}_{\Lambda,\beta} = \int d \phi \,\mbox{e}^{-\beta H^{\rm{cl}}_\Lambda(\phi) }$ is the partition function and $\int d\phi = \int_0^{2\pi}\dots\int_{0}^{2\pi}\prod_{x\in\Lambda}\frac{d\phi_x}{2\pi}$.

We now define the quantum XY model. We restrict ourselves to the spin-$\frac{1}{2}$ case. As before, the model is defined on a finite set of sites $\Lambda$; the Hilbert space is $\caH_\Lambda^\qu = \otimes_{x\in\Lambda} \bbC^2$. The spin operators acting on $\mathbb{C}^{2}$ are the three hermitian matrices $S^i$, $i=1,2,3$, that satisfy $\left[S^1, S^2\right] = i S^3$ and its cyclic permutations, and 
$(S^1)^2+(S^2)^2+(S^3)^2 =~ \tfrac{3}{4} \mathds{1}$. They are explicitly formulated in terms of Pauli matrices:
\be
S^1 = \tfrac{1}{2}\left(\begin{array}{cc}0 &1\\1 &0\end{array}\right),\;\;\;\;S^2 = \tfrac{1}{2}\left(\begin{array}{cc}0 &-i\\i &0\end{array}\right),\;\;\;\;S^3 = \tfrac{1}{2}\left(\begin{array}{cc}1 &0\\0 &-1\end{array}\right).
\ee
The hamiltonian describing the interaction is
\begin{equation}
H^\qu_\Lambda = -\sum_{A\subset \Lambda}J^1_A \prod_{x\in A}S^1_x + J^2_A \prod_{x\in A}S^2_A,\label{quantumXY}
\end{equation}
with $S^i_x =  S^i\otimes \mathds{1}_{\Lambda\backslash \{x\}}$. The $\{J^i_A\}_{A\subset\Lambda}$ are nonnegative coupling constants. The Gibbs state at inverse temperature $\beta$ is
\be
\langle \mathcal{O} \rangle^\qu_{\Lambda,\beta} = \frac{1}{Z^\qu_{\Lambda,\beta}}\Tr \mathcal{O}\e{-\beta H^\qu_\Lambda},
\ee
with $Z^\qu_{\Lambda,\beta} = \Tr \e{-\beta H^\qu_\Lambda}$ the partition function and $\mathcal{O}$ any operator acting on $\mathcal{H}^\qu_\Lambda$.

The first result holds for both classical and quantum models.

\begin{thm} Assume that $J^1_A , J^2_A \geq 0$ for all $A \subset \Lambda$. The following inequalities hold true for all $X,Y \subset \Lambda$, and for all $\beta>0$.
	\begin{align*}
	&Classical:&&\Bigl\langle \prod_{x \in X} \sigma_x^1 \prod_{x \in Y} \sigma_x^1 \Bigr\rangle^\cl_{\Lambda,\beta}- \Bigl\langle \prod_{x \in X} \sigma_x^1 \Bigr\rangle^\cl_{\Lambda,\beta} \; \Bigl\langle \prod_{x \in Y} \sigma_x^1 \Bigr\rangle_{\Lambda,\beta}^\cl \geq 0,\\
	&&&\Bigl\langle \prod_{x \in X} \sigma_x^1 \prod_{x \in Y} \sigma_x^2 \Bigr\rangle^\cl_{\Lambda,\beta} - \Bigl\langle \prod_{x \in X} \sigma_x^1 \Bigr\rangle^\cl_{\Lambda,\beta} \; \Bigl\langle \prod_{x \in Y} \sigma_x^2 \Bigr\rangle^\cl_{\Lambda,\beta} \leq 0.\\
	&Quantum:&&\Bigl\langle \prod_{x \in X} S_x^1 \prod_{x \in Y} S_x^1 \Bigr\rangle^\qu_{\Lambda,\beta} - \Bigl\langle \prod_{x \in X} S_x^1 \Bigr\rangle^\qu_{\Lambda,\beta} \; \Bigl\langle \prod_{x \in Y} S_x^1 \Bigr\rangle^\qu_{\Lambda,\beta} \geq 0, \\
	&&&\Bigl\langle \prod_{x \in X} S_x^1 \prod_{x \in Y} S_x^2 \Bigr\rangle^\qu_{\Lambda,\beta} - \Bigl\langle \prod_{x \in X} S_x^1 \Bigr\rangle^\qu_{\Lambda,\beta} \; \Bigl\langle \prod_{x \in Y} S_x^2 \Bigr\rangle^\qu_{\Lambda,\beta} \leq 0.
	\end{align*}
	\label{XYGriffiths}
\end{thm}

In the quantum case, similar inequalities hold for Schwinger functions, see \cite{BLU} for details.
The proofs are given in Sections \ref{secproofs}, \ref{quantXYproofs} respectively. These inequalities are known as Ginibre inequalities --- first introduced  by Griffiths  for the Ising model \cite{Gri} and systematised in a seminal work by Ginibre \cite{Gin}, which provides a general framework for inequalities of this form. Ginibre inequalities for the classical XY model have then been established with different techniques \cite{Gin, KPV75, Mon, KPV76, MP}. The equivalent result for the quantum case has been proved with different approaches \cite{Gal,Suz,PM2,BLU}. An extension to the ground state of quantum systems with spin 1 was proposed in \cite{BLU}.  A straightforward corollary of this theorem is monotonicity with respect to coupling constants, as we see now.

\begin{cor}
	\label{Correlation monotonicity}
	\label{dS/dJ}
	Assume that $J^1_A , J^2_A \geq 0$ for all $A \subset \Lambda$. Then for all $X,Y \subset \Lambda$, and for all $\beta>0$
	\begin{align*}
	&Classical:&\mbox{\hspace{-2cm}}\frac{\partial}{\partial J^1_Y}\langle \prod_{x\in X}\sigma^1_x \rangle_{\Lambda,\beta}^\cl \geq 0,\\
	&&\mbox{\hspace{-2cm}}\frac{\partial}{\partial J_Y^2}\langle \prod_{x\in X}\sigma^1_x \rangle_{\Lambda,\beta}^\cl \leq 0.\\
	&Quantum:&\mbox{\hspace{-2cm}}\frac{\partial}{\partial J^1_Y}\langle \prod_{x\in X}S^1_x \rangle_{\Lambda,\beta}^\qu \geq 0,\\
	&&\mbox{\hspace{-2cm}}\frac{\partial}{\partial J_Y^2}\langle \prod_{x\in X}S^1_x \rangle_{\Lambda,\beta}^\qu \leq 0.\\
	\end{align*}
\end{cor}

Interestingly this result appears to be not trivially true for the quantum Heisenberg ferromagnet.  Indeed a toy version of the fully SU(2) invariant  model has been provided explicitly, for which this result does not hold  (nearest neighbours interaction on a three-sites chain with open boundary conditions) \cite{HS}. The question whether this result might still be established in a proper setting is still open. On the other hand, Ginibre inequalities have been proved for the classical Heisenberg ferromagnet \cite{KPV75,Dun,  MP}.

Monotonicity of correlations with respect to temperature does not follow straightforwardly from the corollary. This can nonetheless be proved for the classical XY model.

\begin{thm}\hfill
	
	Classical model: Assume that $J^1_A\geq|J^2_A|$ for all $A \subset \Lambda$, and that $J^2_A=0$ whenever $|A|$ is odd. Then for all $A,B \subset \Lambda$, we have
	\be
	\frac\partial{\partial \beta} \Bigl\langle \prod_{x \in B} \sigma_x^1 \Bigr\rangle^{\cl}_{\Lambda,\beta} \geq 0.
	\nonumber
	\ee
	Let us restrict to the two-body case and assume that $H^{\cl}_\Lambda$ is given by
	\be
	H^{\cl}_\Lambda = -\sum_{x,y\in\Lambda}J_{xy}\left(\sigma^1_x\sigma^1_y + \eta_{xy}\sigma^2_x\sigma^2_y\right).\label{ham2body}\nonumber\ee
	Then if $|\eta_{xy}|\leq 1$ for all $x,y$,
	\be
	\frac{\partial}{\partial J_{xy}} \langle\prod_{z\in A}\sigma^1_z\rangle^\cl_{\Lambda,\beta}\geq 0.
	\ee
	\label{monotoneXY}
\end{thm}
This result has been proposed and discussed in various works \cite{Gin, MP,MMP} --- see Section \ref{secproofs} for the details. Unfortunately we lack a quantum equivalent of these statements.

We conclude this section by remarking that correlation inequalities in the quantum case can be applied also to other models of interest.  For example, we consider a certain formulation of Kitaev's model (see \cite{Kit} for its original formulation and \cite{Bach} for a review of the topic). Let $\Lambda \subset \subset \bbZ^2$ be a square lattice with edges $\mathcal{E}_\Lambda$. Each {\itshape edge} of the lattice hosts a spin, i.e. the hilbert space of this model is $\caH^{\rm{Kitaev}}_\Lambda = \otimes_{e\in\mathcal{E}_\Lambda}\bbC^2$.  The Kitaev hamiltonian is
\be
H^{{\rm Kitaev}}_\Lambda = -\sum_{x\in\Lambda} J_x\prod_{\substack{e\in\mathcal{E}_\Lambda:\\x\in e}}S^1_e -\sum_{F\subset\Lambda} J_{F} \prod_{e\subset F} S^3_e,
\ee
where $F$ denotes the faces of the lattice, i.e the unit squares which are the building blocks of $\bbZ^2$, $J_x,\, J_F$ are ferromagnetic coupling constants and $S^i_e = S^i\otimes\mathds{1}_{\mathcal{E}_\Lambda \backslash e}$. $H^{{\rm Kitaev}}_\Lambda$  has the same structure as hamiltonian \eqref{quantumXY} so Ginibre inequalities apply as well. It is not clear, though, whether this might lead to useful results for the study of this specific model.

Another relevant model is the {\it plaquette orbital model} that was studied in \cite{WJ,BK};
interactions between neighbours $x,y$ are of the form $-S_x^i S_y^i$, with $i$ being
equal to 1 or 3 depending on the edge.

\section{Comparison between Ising and XY models}

We now compare the correlations of the Ising and XY models and their respective critical temperatures. The configuration space of the Ising model is $\Omega_\Lambda^{\rm{Is}} = \{-1,1\}^\Lambda$, that is, Ising configurations are given by $\{s_x\}_{x\in\Lambda}$ with $s_x=\pm1$ for each $x\in\Lambda$. We consider many-body interactions, so the energy of a configuration $s\in\Omega_\Lambda^{\rm{Is}}$ is
\be
H_{\Lambda,\{J_A\}}^{\rm{Is}}(s) = -\sum_{A\subset\Lambda} J_A\prod_{x\in A} s_x;
\label{Ising}
\ee
we assume that the system is ``ferromagnetic", i.e.\ the coupling constants $J_A\geq 0$ are nonnegative. The Gibbs state at inverse temperature $\beta$ is
\be
\langle f \rangle^{\rm{Is}}_{\Lambda,\{J_A\},\beta} = \frac{1}{Z_{\Lambda,\{J_A\},\beta}^{\rm{Is}}}\sum _{s\in\Omega^{\rm{Is}}_\Lambda}f(s) \e{-\beta H^{\rm{Is}}_{\Lambda,\{J_A\}}},
\ee
with $f$ any functional on $\Omega_\Lambda^{\rm{Is}}$ and $Z_{\Lambda,\{J_A\},\beta}^{\rm{Is}} = \sum_{s\in\Omega^{\rm{Is}}_{\Lambda}} \e{-\beta H^{\rm{Is}}_{\Lambda,\{J_A\}}}$ is the partition function.
The following result holds for both the classical \cite{KPV76} and the quantum case \cite{Suz,Pe}.

\begin{thm}\label{XYvsIsing} Assume that $J^1_A,J^2_A\geq 0$ for all $A\subset\Lambda$. Then for all $X\subset\Lambda$ and all $\beta>0$,
	\[
	\begin{split}
	\text{Classical:} \qquad &\langle \prod_{x\in X}\sigma^1_x\rangle^{\cl}_{\Lambda,\beta} \leq \langle \prod_{x\in X}s_x\rangle^{{\rm Is}}_{\Lambda,\{J^1_A\},\beta}.\\
	\text{Quantum:}  \qquad &\langle \prod_{x\in X}S^1_x\rangle^{\qu}_{\Lambda,\beta} \leq 2^{-|X|}\langle \prod_{x\in X}s_x\rangle^{\rm{Is}}_{\Lambda,\{J^*_A\},\beta},
	\end{split}
	\]
	with $J^*_A = 2^{-|A|}J^1_A$.
\end{thm}

A review of the proof of the classical case is proposed in Section \ref{secproofs}.
In the quantum case, this statement for spin-$\frac{1}{2}$ is a straightforward consequence of Corollary \ref{dS/dJ}, but interestingly this result has been extended to any value of the spin \cite{Pe}. We review the proof of this general case in Section \ref{quantXYproofs}.

We now consider the case of spin-$\frac12$ and pair interactions, that is, the hamiltonian is
\be
H_{\Lambda}^{\qu}=-\sum_{ x,y\in\Lambda}(S_x^1S_y^1+ S_x^2S_y^2).
\ee

We define the \emph{spontaneous magnetisation} $m(\beta)$ at inverse temperature $\beta$ by
\be
m(\beta)^2 =\liminf_{\Lambda \nearrow \bbZ^d}\frac{1}{|\Lambda|^2}\sum_{x,y\in\Lambda}\langle S_x^1S_y^1\rangle_{\Lambda,\beta}^\qu.
\ee
We define the critical temperature for the model $T_{\rm c}^\qu = 1/\beta_{\rm c}^\qu$ as
\be
\beta_{\rm c}^\qu = \sup \bigl\{ \beta>0 : m(\beta)=0 \bigr\},
\ee
where $\beta_{\rm c}^\qu \in (0,\infty]$. A consequence of Theorem \ref{XYvsIsing} is the following.

\begin{cor}
	\label{XyIsingCritTemp}
	The critical temperatures satisfy
	\[
	T_{\rm c}^\qu \leq \tfrac14 T_{\rm c}^{\rm Ising}.
	\]
\end{cor}

The critical temperature of the Ising model in the three-dimensional cubic lattice has been calculated numerically and is $T_{\rm c}^{\rm Ising}=4.511\pm0.001$ \cite{TVPS}. It is $T_{\rm c}^{\rm cl}=2.202\pm0.001$ \cite{HM} for the classical model and $T_{\rm c}^\qu = 1.008\pm 0.001$ for the quantum model \cite{Wes}.

A major result of mathematical physics is the rigorous proof of the occurrence of long-range order in the classical and quantum XY models, in dimensions three and higher, and if the temperature is low enough \cite{FSS,DLS}. The method can be used to provide a rigorous lower bound on critical temperatures; the following theorem concerns the quantum model.

\begin{thm}\label{XyCritTemp}
	For the three-dimensional cubic lattice, the temperature of the quantum XY model satisfies
	\[
	T_{\rm c}^\qu \geq 0.323.
	\]
\end{thm}

The best rigorous upper bound on the critical temperature of the three-dimensional Ising model is $T_c^{\rm{Ising}}\leq 5.0010$ \cite{Vig}. Together with the above corollary and theorem, we get
\be
0.323\leq T_{\rm c}^\qu \leq \tfrac14 T_{\rm c}^{\rm Ising} \leq 1.250.
\ee

\begin{proof}[Proof of Theorem \ref{XyCritTemp}]
	We consider the XY model with spins in the 1-3 directions for convenience. We make use of the result \cite[Theorem 5.1]{U}, that was obtained with the method of reflection positivity and infrared bounds \cite{FSS,DLS}. Precisely, we use Equations (5.54), (5.57) and (5.63) of \cite{U}.
	\be
	m(\beta)^2 \geq 
	\begin{cases}
		\frac{1}{4}-\frac{J_3}{2}\sqrt{\langle S^1_0S^1_{e_1}\rangle^\qu}-\frac{K_3}{\beta}
		\\
		\langle S^1_0S^1_{e_1}\rangle^\qu_{{\Lambda}}-\frac{I_3}{2}\sqrt{\langle S^1_0S^1_{e_1}\rangle^\qu}-\frac{K_3'}{\beta}
		\label{XYIRB}
	\end{cases}
	\ee
	where $e_1$ is a nearest neighbour of the origin, and $J_3,I_3,K_3,K_3'$ are real numbers coming from explicit integrals. Their values are $J_3=1.15672$; $I_3=0.349884$; $K_3=0.252731$; and $K_3'=0.105107$. Notice that $\beta$ is rescaled by a factor 2 with respect to \cite{U}, due to a different choice of coupling constants in the hamiltonian. Let $x = \sqrt{\langle S^1_0S^1_{e_1}\rangle^\qu}$; since we do not have good bounds on $x$, we treat it as an unknown. The magnetisation $m(\beta)$ is guaranteed to be positive if $x \leq t$ where $t$ is the zero of $\frac14 - \frac{K_3}{\beta} - \frac{J_3}{2}x$; or $x \geq r_+$, where $r_+$ is the largest zero of $x^2 - \frac{I_3}{2} x - \frac{K_3'}{\beta}$. At least one of these holds true when $r_+ < t$, that is, when
	\be
	\tfrac{1}{2} \Bigl( \tfrac{I_3}{2} + \sqrt{\tfrac{I_3^2}{4}+\tfrac{4K_3'}{\beta}} \Bigr) < \tfrac{1}{2J_3} - \tfrac{1}{\beta} \tfrac{2K_3}{J_3}
	\ee
	This is the case for $1/\beta < 0.323$ giving the upper bound $T_c\geq 0.323$.
\end{proof}
\section{Proofs for the classical XY model}\label{secproofs}

The proofs require several steps and additional lemmas. The following paragraphs are devoted to a complete study of their proofs. Given local variables $\{\sigma_x\}_{x\in\Lambda}$, we denote $\sigma^i_A =\prod_{x\in A}\sigma^i_x$ for $A\subset\Lambda$.

\subsection{Griffiths and FKG inequalities, and proof of Theorem \ref{XYGriffiths}}

We start with Theorem \ref{XYGriffiths}. We describe the approach proposed in \cite{KPV75, KPV76}, and use a similar notation. Their framework relies on some well known properties of the Ising model and on the so called FKG inequality.

\begin{lem}[Griffiths inequalities for the Ising model]
	Let $f$ and $g$ be functionals on $\Omega_\Lambda^{\rm{Is}}$ such that they can be expressed as power series of $\prod_{x\in A} s_x$, $A\subset\Lambda$ with positive coefficients. Then
	\[
	\begin{split}
	\langle f \rangle_{\Lambda,\{J_A\},\beta}^{\rm Is} &\geq0;\\
	\langle fg \rangle_{\Lambda,\{J_A\},\beta}^{\rm Is} &\geq\langle f \rangle_{\Lambda,\{J_A\},\beta}^{\rm Is}\langle g \rangle_{\Lambda,\{J_A\},\beta}^{\rm Is}.
	\end{split}
	\]
	\label{IsingGriffiths}
\end{lem}

We do not provide the proof of this result --- see \cite{Gri,Gin} for the original formulation and \cite{FV} for a modern description. An immediate consequence is the following.

\begin{cor} Given $f$ with the properties in Lemma \ref{IsingGriffiths}, we have for any $A\subset \Lambda$
	\[
	\frac{\partial}{\partial J_A}\langle f\rangle_{\Lambda,\{J_A\},\beta}^{\rm Is} \geq 0.
	\]\label{Isingmonotone}
\end{cor}

Another result which is very useful in this framework is the so called FKG inequality. We formulate it in a specific setting. Let $\mathcal{I}_N = \left[0,\frac{\pi}{2}\right]^N$ for some $N\in\mathbb{N}$. Any $\psi\in\mathcal{I}_N$ is then a collection of angles $\psi = \left(\psi_1,\dots,\psi_N\right)$. It is possible to introduce a partial ordering relation on $\mathcal{I}_N$ as follows: for any $\psi,\xi\in~\mathcal{I}_N$, $\psi\leq \xi$ if and only if $\psi_i\leq\xi_i$ for all $i\in\{1,\dots,N\}$. A function $f$ on $\mathcal{I}_N$ is said to be increasing (or decreasing) if $\psi\leq\xi$ implies $f(\psi)\leq ~f(\xi)$ (or $f(\psi)\geq f(\xi)$) for all $\psi,\xi \in \mathcal{I}_N$. The following result holds.

\begin{lem}[FKG inequality]
	Let $d\nu(\psi) = p(\psi) \prod_{i=1}^N d\mu(\psi_i)$ be a normalised measure on $\mathcal{I}_N$, with $d\mu(\psi_i)$ a normalised measure on $\left[0,\frac{\pi}{2}\right]$, $p(\psi)\geq 0$ for all $\psi\in\mathcal{I}_N$ and
	\be
	p(\psi\vee\xi)p(\psi\wedge\xi)\geq p(\psi)p(\xi), \label{hypmeasure}
	\ee
	where $(\psi\vee\xi)_i= \max (\psi_i,\xi_i)$ and $(\psi\wedge\xi)_i= \min (\psi_i,\xi_i)$. Then for any $f$ and $g$ increasing (or decreasing) functions on $\mathcal{I}_N$
	\be
	\int fgd\nu \geq \int f d\nu \int g d\nu.\nonumber
	\ee
	The inequality changes sign if one of the functions is increasing and the other is decreasing.
	\label{FKG}
\end{lem}

We also skip the proof of this statement. We refer to \cite{FKG} for the original result, to \cite{Pr,KPV76} for the formulation above, and \cite{FV} for its relevance in the study of the Ising model.

Before turning to the actual proof of the theorem, we introduce another useful lemma.
\begin{lem}
	Let $\{q_x\}_{x\in\Lambda}$ be a collection of positive increasing (decreasing) functions on $\left[0,\frac{\pi}{2}\right]$. Then for any $\theta,\psi\in\mathcal{I}_{|\Lambda|}$ and any $A\subset\Lambda$,
	\[
	q_A(\theta\vee\psi)+q_A(\theta\wedge\psi) \geq q_A(\psi)+q_A(\theta).
	\nonumber \]\label{q_Asum}
\end{lem}
We do not provide the proof here, see \cite{Pr,KPV76} for more details.
We can now discuss the proof of Theorem \ref{XYGriffiths}.

\begin{proof}[Proof of Theorem \ref{XYGriffiths}]
	Since the temperature does not play any r\^ole in this section, we set $\beta=1$ in the following and we drop any dependency on it.
	The main idea of the proof is to describe a classical XY spin as a pair of Ising spins and an angular variable. The new notation for $\sigma_x\in S^1$ is
	\begin{eqnarray}
	\sigma_x^1 &=& \cos(\theta_x)U_x, \label{sigma1} \\
	\sigma_x^2&=& \sin(\theta_x)V_x,\label{sigma2}
	\end{eqnarray}
	with $U_x,V_x \in \{-1,1\}$ for all $x\in\Lambda$  and $\theta = \left(\theta_{x_1},\dots,\theta_{x_\Lambda}\right)\in\mathcal{I}_{|\Lambda|}$.
	With this notation, it is possible to express $H^{\rm cl}_\Lambda$ of Eq. \ \eqref{classicalXY} as the sum of two Ising hamiltonians with spins $\{U_x\}_{x\in\Lambda}$, $\{V_x\}_{x\in\Lambda}$ respectively:
	\begin{eqnarray}
	H^{\cl}_{\Lambda}\left(\theta, U,V\right) &=& -\sum_{A\subset\Lambda}\left(J^1_A\prod_{x\in A}\cos(\theta_x) U_A +J^2_A \prod_{x\in A}\sin(\theta_x)V_A\right)\\
	&=& H^{\rm Is}_{\Lambda,\{\cos(\theta)_AJ^1_A\}}(U) +H^{\rm Is}_{\Lambda,\{\sin(\theta)_AJ^2_A\}}(V).\label{XYbyIsing}
	\end{eqnarray} 
	Let us introduce the notation: $J^1_A\prod_{x\in A}\cos(\theta_x) = \mathscr{J}_A(\theta)$, $J^2_A\prod_{x\in A}\sin(\theta_x) = \mathscr{K}_A(\theta)$, $\int d\theta = \int_0^{\frac{\pi}{2}}\dots\int_{0}^{\frac{\pi}{2}}\prod_{x\in\Lambda}\frac{2}{\pi}d\theta_x$. Then
	\begin{eqnarray}
	&&\langle \sigma^1_X \sigma^1_Y\rangle_{H^{\cl}_\Lambda} = \frac{\int d\theta\,Z^{\rm{Is}}_{\Lambda,\{ \mathscr{J}_A(\theta)\}}Z^{\rm{Is}}_{\Lambda,\{\mathscr{K}_A(\theta)\}}\cos(\theta)_X\cos(\theta)_Y\langle U_X U_Y\rangle^{{\rm Is}}_{\Lambda,\{\mathscr{J}_A(\theta)\}}}{\int d\theta\, Z^{\rm{Is}}_{\Lambda,\{\mathscr{J}_A(\theta)\}}Z^{\rm{Is}}_{\Lambda,\{\mathscr{K}_A(\theta)\}}}\nonumber\\
	&&\geq\frac{\int d\theta\,Z^{\rm{Is}}_{\Lambda,\{\mathscr{J}_A(\theta)\}}Z^{\rm{Is}}_{\Lambda,\{\mathscr{K}_A(\theta)\}}\cos(\theta)_{X}\langle U_X \rangle^{{\rm Is}}_{\Lambda,\{\mathscr{J}_A(\theta)\}}\cos(\theta)_{Y}\langle U_Y\rangle^{{\rm Is}}_{\Lambda,\{\mathscr{J}_A(\theta)\}}}{\int d\theta \, Z^{\rm{Is}}_{\Lambda,\{\mathscr{J}_A(\theta)\}}Z^{\rm{Is}}_{\Lambda,\{\mathscr{K}_A(\theta)\}}}.\nonumber
	\end{eqnarray}
	\\
	The inequality above follows from Lemma \ref{IsingGriffiths}.
	Moreover
	\be
	\langle \sigma^1_X \sigma^2_Y\rangle_{H^{\cl}_\Lambda} = \frac{\int d\theta Z^{\rm{Is}}_{\Lambda,\{\mathscr{J}_A(\theta)\}}Z^{\rm{Is}}_{\Lambda,\{\mathscr{K}_A(\theta)\}}\cos(\theta)_X\langle U_X\rangle^{\rm{Is}}_{\Lambda,\{\mathscr{J}_A(\theta)\}}\sin(\theta)_Y \langle V_Y\rangle^{\rm Is}_{\Lambda,\{\mathscr{K}_A(\theta)\}}}{\int d\theta\, Z^{\rm{Is}}_{\Lambda,\{\mathscr{J}_A(\theta)\}}Z^{\rm{Is}}_{\Lambda,\{\mathscr{K}_A(\theta)\}}}.\nonumber
	\ee
	\\
	$\cos(\theta)_X$ and $\sin(\theta)_X$ are respectively decreasing and increasing on $\mathcal{I}_{|\Lambda|}$ for any $X\subset\Lambda$. Let us now consider $\langle U_X\rangle_{\Lambda,\{\mathscr{J}_A(\theta)\}}^{\rm Is}$. By Corollary \ref{Isingmonotone}, it is a decreasing function on $\mathcal{I}_{|\Lambda|}$ for any $X\subset~\Lambda$, since the coupling constants of $H^I_{\Lambda,\{\mathscr{J}_A(\theta)\}}$ are decreasing in $\theta$. Analogously, $\langle V_X\rangle_{\Lambda,\{\mathscr{K}_A(\theta)\}}^{\rm Is}$ is an increasing function on $\mathcal{I}_{|\Lambda|}$ for any $X\subset~\Lambda$. Theorem \ref{XYGriffiths} is then a simple consequence of Lemma \ref{FKG}, with $d\mu(\theta_x) =~\frac{2}{\pi}d\theta_x$ and 
	\be
	p(\theta) =\frac{Z^{\rm{Is}}_{\Lambda,\{\mathscr{J}_A(\theta)\}}Z^{\rm{Is}}_{\Lambda,\{\mathscr{K}_A(\theta)\}}}{\int d\theta\,Z^{\rm{Is}}_{\Lambda,\{\mathscr{J}_A(\theta)\}}Z^{\rm{Is}}_{\Lambda,\{\mathscr{K}_A(\theta)\}}}.
	\ee
	
	The last step missing is to show that $p(\theta)$ defined as above fulfills hypothesis \eqref{hypmeasure} of Lemma \ref{FKG}. This amounts to showing
	\begin{eqnarray}
	Z^{\rm{Is}}_{\Lambda,\{\mathscr{K}_A(\theta\vee\psi)\}}Z^{\rm{Is}}_{\Lambda,\{\mathscr{K}_A(\theta\wedge\psi)\}}&\geq& Z^{\rm{Is}}_{\Lambda,\{\mathscr{K}_A(\theta)\}}Z^{\rm{Is}}_{\Lambda,\{\mathscr{K}_A(\psi)\}};\label{pZsin}\\
	Z^{\rm{Is}}_{\Lambda,\{\mathscr{J}_A(\theta\vee\psi)\}}Z^{\rm{I}}_{\Lambda,\{\mathscr{J}_A(\theta\wedge\psi)\}}&\geq& Z^{\rm Is}_{\Lambda,\{\mathscr{J}_A(\theta)\}}Z^{\rm{Is}}_{\Lambda,\{\mathscr{J}_A(\psi)\}}.\label{pZcos}
	\end{eqnarray}
	Since the arguments to prove these inequalities are very similar, we prove explicitly only the first one. Eq.\ \eqref{pZsin} is equivalent to
	\be
	\left(\frac{Z^{\rm{Is}}_{\Lambda,\{\mathscr{K}_A(\theta)\}}}{Z^{\rm{Is}}_{\Lambda,\{\mathscr{K}_A(\theta\wedge\psi)\}}}\right)^{-1}
	\left(\frac{Z^{\rm{Is}}_{\Lambda,\{\mathscr{K}_A(\theta\vee\psi)\}}}{Z^{\rm{Is}}_{\Lambda,\{\mathscr{K}_A(\psi)\}}}\right) \geq 1
	\ee 
	Notice that
	\be
	\left(\frac{Z^{\rm{Is}}_{\Lambda,\{\mathscr{K}_A(\theta)\}}}{Z^{\rm{Is}}_{\Lambda,\{\mathscr{K}_A(\theta\wedge\psi)\}}}\right)^{-1}
	\left(\frac{Z^{\rm{Is}}_{\Lambda,\{\mathscr{K}_A(\theta\vee\psi)\}}}{Z^{\rm{Is}}_{\Lambda,\{\mathscr{K}_A(\psi)\}}}\right)
	= \frac{\langle\e{-H^{\rm{Is}}_{\Lambda,\{\mathscr{K}_A(\theta\vee\psi)-\mathscr{K}_A(\psi)\}}}\rangle_{\Lambda, \{\mathscr{K}_A(\psi)\}}^{\rm Is}}{\langle \e{-H^{\rm{Is}}_{\Lambda,\{\mathscr{K}_A(\theta)-\mathscr{K}_A(\theta\wedge\psi)\}}}\rangle_{\Lambda, \{\mathscr{K}_A(\theta\wedge\psi)\}}^{\rm Is}}. \nonumber
	\ee
	Thanks to Lemma \ref{q_Asum}, the functions whose expectation value we are computing above fulfill the hypothesis of Lemma \ref{IsingGriffiths} and Corollary \ref{Isingmonotone}. 
	Then, applying Lemma \ref{q_Asum} and Corollary \ref{Isingmonotone},
	\be
	\begin{split}
		\langle\e{-H^{\rm{Is}}_{\Lambda,\{\mathscr{K}_A(\theta\vee\psi)-\mathscr{K}_A(\psi)\}}}\rangle_{\Lambda, \{\mathscr{K}_A(\psi)\}}^{\rm Is}
		&\geq \langle\e{-H^{\rm{Is}}_{\Lambda,\{\mathscr{K}_A(\theta)-\mathscr{K}_A(\theta\wedge\psi)\}}}\rangle_{\Lambda, \{\mathscr{K}_A(\psi)\}}^{\rm Is} \\
		&\geq \langle\e{-H^{\rm{Is}}_{\Lambda,\{\mathscr{K}_A(\theta)-\mathscr{K}_A(\theta\wedge\psi)\}}}\rangle_{\Lambda, \{\mathscr{K}_A(\theta\wedge\psi)\}}^{\rm Is}.
	\end{split}
	\ee
	Hence $p(\theta)$ has the required property.
\end{proof}

\subsection{Proof of Theorem \ref{monotoneXY}}\label{ThmmonotoneXY}

Let us now turn to Theorem \ref{monotoneXY}. We follow the framework described in \cite{Gin,MP}.

\begin{lem}
	Let $H^\cl_\Lambda$ be the hamiltonian defined in \eqref{classicalXY}. If $J^1_A\geq|J^2_A|$ for all $A \subset \Lambda$ and $J^2_A = 0$ for $|A|$ odd, then there exist non negative coupling constants $\{K_M\}_{M\in\bbZ^\Lambda}$ such that
	\begin{equation}
	H^\cl_\Lambda(\phi) = -\sum_{M\in \bbZ^\Lambda}K_M \cos\left(M\cdot \mathbb{\phi}\right),\end{equation}	
	where, given $M\in\bbZ^\Lambda$, $M = (m_1, m_2, \dots, m_\Lambda)$, $M\cdot\mathbb{\phi}=\sum_{x\in\Lambda}m_x\phi_x$.
	\label{newform}
\end{lem}
\begin{proof}[Proof of Lemma \ref{newform}]
	The statement follows from the two following identities:
	\begin{eqnarray}
	\cos(\theta)\cos(\chi) &=& \frac{1}{2}(\cos(\theta-\chi)+\cos(\theta+\chi)),\\
	\sin(\theta)\sin(\chi) &=& \frac{1}{2}(\cos(\theta-\chi)-\cos(\theta+\chi)),
	\end{eqnarray}
	$\forall \theta, \chi \in \left[0,2\pi\right]$. 
\end{proof}
A necessary step for this lemma and for Theorem \ref{monotoneXY} is duplication of variables \cite{Gin}: we consider two sets of angles (i.e. spins) on the lattice instead of just one, and denote them by $\{\phi_x\}_{x\in\Lambda}$ and $\{\bar{\phi}_x\}_{x\in\Lambda}$. The hamiltonian for the $\{\bar{\phi}_x\}$ is
\begin{eqnarray}
\bar{H}^\cl_\Lambda(\bar{\phi}) = -\sum_{A\subset\Lambda}\left(\bar{J}^1_A\bar{\sigma}^1_A+\bar{J}^2_A\bar{\sigma}^2_A\right)\nonumber\\
= -\sum_{M\in\bbZ^\Lambda}\bar{K}_M\cos(M\cdot\bar{\phi}).
\end{eqnarray}
Here, $\{\bar{\sigma}_x\}$ are related to  $\{\bar{\phi}_x\}$ as in Eq.s \eqref{cos} and \eqref{sin}. The $\bar{J}^i_A$ are non negative coupling constants with $\bar{J}^1_A\geq|\bar{J}^2_A|\geq 0$ and $\{\bar{K}_M\}$ are as in Lemma \ref{newform}.
A composite hamiltonian can be defined as
\be
\begin{split}
	&-\hat{H}_\Lambda (\phi,\bar{\phi}) = -H^\cl_\Lambda(\phi) - \bar{H}^\cl_\Lambda(\bar{\phi})\\
	&=\sum_{M\in\mathcal{M}} \tfrac{K_M +\bar{K}_M}{2} \left(\cos(M\cdot\phi)
	+\cos(M\cdot\bar{\phi})\right) + \tfrac{K_M -\bar{K}_M}{2} \left(\cos(M\cdot\phi)-\cos(M\cdot\bar{\phi})\right)
\end{split}
\ee
In the following we always suppose $K_M\geq\bar{K}_M$ for all $M\in\bbZ^\Lambda$.The expectation value of any functional  $f(\phi,\bar{\phi})$ can be written as 
\begin{equation}
\langle f \rangle_{\hat{H}_\Lambda,\beta} = \frac{1}{Z_{H_\Lambda,\beta}Z_{\bar{H}_\Lambda,\beta}}\int d\phi \,d\bar{\phi} \mbox{e}^{-\beta \hat{H}_\Lambda (\phi,\bar{\phi})} f(\phi,\bar{\phi}).
\end{equation}

\begin{lem}
	Suppose $f(\phi,\bar{\phi})$ belongs to the cone generated by $\cos(M\cdot\phi)\pm\cos(M\cdot~\bar{\phi})$, $M\in\bbZ^\Lambda$, i.e. f can be written as product, sum or multiplication by a positive scalar of objects of that form. Then \begin{equation} \langle f \rangle_{\hat{H}_\Lambda,\beta}\geq0.\end{equation}\label{positivecorr}
\end{lem}
\begin{proof}[Proof of Lemma \ref{positivecorr}]
	Firstly, notice that
	\begin{equation}
	\int d\phi\, d\bar{\phi} \,\prod_{s=1}^n (\cos(M_s\cdot\phi)\pm\cos(M_s\cdot\bar{\phi})\geq 0
	\label{integralcos}
	\end{equation}
	for any $M_1,\dots,M_n\in\bbZ^\Lambda$ and any sequence of $(\pm)$. 
	This follows from
	\begin{eqnarray}
	\cos(M\cdot \phi)+ \cos(M\cdot\bar{\phi}) &=& 2 \cos(M\cdot\Phi)\cos(M\cdot\bar{\Phi}),\\
	\cos(M\cdot \phi)- \cos(M\cdot\bar{\phi}) &=& 2 \sin(M\cdot\Phi)\sin(M\cdot\bar{\Phi}),
	\end{eqnarray}
	with $\Phi_i = \frac{1}{2}(\phi_i +\bar{\phi}_i)$ and  $\bar{\Phi}_i = \frac{1}{2}(\phi_i -\bar{\phi}_i)$. The integral \eqref{integralcos} can be formulated as
	\begin{equation}
	\int d\Phi \,d\bar{\Phi} F(\Phi)F(\bar{\Phi}) = \left(\int d\Phi F(\Phi)\right)^2 \geq 0,
	\end{equation}
	with $F(\Phi)$ an appropriate product of sines, cosines and positive constants.
	
	Let us now turn to $ \langle f \rangle_{\hat{H}_\Lambda,\beta}$. Since the partition function is always positive, we can focus on
	\begin{equation}\int d\phi \,d\bar{\phi}\, \mbox{e}^{-\beta \hat{H}_\Lambda (\phi,\bar{\phi})} f(\phi,\bar{\phi}).
	\end{equation}
	By a  Taylor expansion of $\mbox{e}^{-\beta \hat{H}_\Lambda (\phi,\bar{\phi})}$ and by the properties of $f$, this can be expressed as a sum with positive coefficients of integrals in the form \eqref{integralcos}. Hence the nonnegativity of the expectation value.
\end{proof}
We have now all we need to prove Theorem \ref{monotoneXY}.
\begin{proof}[Theorem \ref{monotoneXY}]
	In order to prove the first statement of the theorem we use the formulation of the hamiltonian decribed in Lemma \ref{newform}. Moreover, since $\sigma^1_A$ can be clearly expressed as the sum (with positive coefficients) of terms of the form $\cos(M\cdot\phi)$, $M\in\bbZ^\Lambda$, it is enough to prove that for any $M,N\in\bbZ^\Lambda$
	\be
	\label{2ndgriffithscos}
	\begin{split}
		&\frac{\partial}{\partial K_N}\langle \cos(M\cdot\phi)\rangle^\cl_{\Lambda,\beta} \\
		&= \langle \cos(M\cdot\phi)\cos(N\cdot\phi)\rangle^\cl_{\Lambda,\beta} - \langle \cos(M\cdot\phi)\rangle^\cl_{\Lambda,\beta}\langle\cos(N\cdot\phi)\rangle^\cl_{\Lambda,\beta}\geq 0. \end{split}
	\ee
	
	Consider now the hamiltonian $\hat{H}_\Lambda$ introduced above and $\langle\cdot\rangle_{\hat{H}_\Lambda,\beta}$ the corresponding Gibbs state. From Lemma \ref{positivecorr} we have
	\begin{equation}
	\langle\left(\cos(M\cdot\phi)-\cos(M\cdot\bar{\phi})\right)\left(\cos(N\cdot\phi)-\cos(N\cdot\bar{\phi})\right)\rangle_{\hat{H}_\Lambda,\beta}\geq 0.
	\end{equation}
	If we take the limit $\bar{K}_M 	\nearrow K_M$, we find twice the expression in Eq.\ \eqref{2ndgriffithscos}. Hence the result.
	
	Let us now turn to the second statement of the theorem. In the case of two-body interaction $H^{\cl}_\Lambda$ assumes the form \eqref{ham2body}, which, with a notation resembling the one introduced in Lemma \ref{newform} can be explicitly formulated as
	\be
	H^{\cl}_\Lambda(\phi) = -\sum_{x,y\in\Lambda}K^-_{xy}\cos(\phi_x - \phi_y)+ K^+_{xy}\cos(\phi_x+\phi_y)
	\ee
	with
	\be
	K^\pm_{xy}=\frac{J_{xy}}{2}\left(1\mp\eta_{xy}\right).
	\ee
	Cleary $K^\pm_{xy}$ is analogous to the $K_M$ introduced in Lemma \ref{newform} for $M\in\bbZ^\Lambda$ such that all its elements are zero except $m_x =1,\,m_y=\pm1$.
	Then we have
	\be
	\frac{\partial}{\partial J_{xy}}\langle \sigma_A\rangle_{H^\cl_\Lambda} = \frac{1+\eta_{xy}}{2}\frac{\partial}{\partial K_{xy}^-}\langle \sigma_A\rangle_{H^\cl_\Lambda}+\frac{1-\eta_{xy}}{2}\frac{\partial}{\partial K_{xy}^+}\langle \sigma_A\rangle_{H^\cl_\Lambda}.
	\ee
	Due to Eq.\ \eqref{2ndgriffithscos} the expression above is the sum of two positive terms, hence it is positive.
	
\end{proof}

\subsection{Proof of Theorem \ref{XYvsIsing}}
In this section we discuss the proof of Theorem \ref{XYvsIsing} for the classical XY model. We use some of the concepts introduced in Section \ref{ThmmonotoneXY}. The present proof has ben proposed in \cite{KPV75,KPV76}.
\begin{proof}[Proof of Theorem \ref{XYvsIsing}]
	As for the proof of Theorem \ref{XYGriffiths}, we express the XY spins by means of two Ising spins and an angle in $\left[0,\frac{\pi}{2}\right]$ - see Eq.s\ \eqref{sigma1},\ \eqref{sigma2} for the explicit expression of the spins and \eqref{XYbyIsing} for the new formulation of the hamiltonian $H^{\cl}_\Lambda$.
	With the same notation:
	\be
	\begin{split}
		\langle \sigma^1_X\rangle_{H^\cl_\Lambda} &= \frac{\int d\theta\,Z_{\Lambda, \{\mathscr{J}_A(\theta)\}}^{\rm{Is}}Z_{\Lambda, \{\mathscr{K}_A(\theta)\}}^{\rm{Is}}\cos(\theta)_X\langle U_X\rangle_{\Lambda,\{\mathscr{J}_A(\theta)\}}^{\rm Is}}{\int d \theta\,Z_{\Lambda, \{\mathscr{J}_A(\theta)\}}^{\rm{Is}}Z_{\Lambda, \{\mathscr{K}_A(\theta)\}}^{\rm{Is}}} \\
		&\leq \frac{\int d\theta\,Z_{\Lambda, \{\mathscr{J}_A(\theta)\}}^{\rm{Is}}Z_{\Lambda, \{\mathscr{K}_A(\theta)\}}^{\rm{Is}}\max_{\theta\in\mathcal{I}_{|\Lambda|}}\langle U_X\rangle_{\Lambda,\{\mathscr{J}_A(\theta)\}}^{\rm Is}}{\int d\theta\,Z_{\Lambda, \{\mathscr{J}_A(\theta)\}}^{\rm{Is}}Z_{\Lambda, \{\mathscr{K}_A(\theta)\}}^{\rm{Is}}} \\
		&= \langle U_A\rangle_{\Lambda,\{J^1_A\}}^{\rm Is}.
	\end{split}
	\ee
\end{proof}

\section{Proofs for the quantum XY model}
\label{quantXYproofs}

We now discuss the proof of Theorem \ref{XYGriffiths} in the quantum case. This theorem has been proved for pair interaction in \cite{Gal}, and it has been proposed independently in various works for more generic interactions \cite{Suz, PM2,BLU}. We describe here the simpler approach proposed in \cite{BLU}.
Since the temperature does not play any role from now on, we set $\beta=1$ and omit any dependency on it in the following. As for the classical case we introduce the notation $S^i_A = \prod_{x\in A}S^i_x $.

\begin{proof}[Proof of Theorem \ref{XYGriffiths}]
	For the proof it is convenient to perform a unitary transformation on the hamiltonian \eqref{quantumXY} and consider its version with interactions along the first and third directions of spin, namely
	\be
	H^\qu_\Lambda = -\sum_{A\subset\Lambda}J^1_A S^1_A + J^3_A S^3_A,\label{quantumXY13}
	\ee
	with $J^3_A = J^2_A$ for all $A\subset\Lambda$.
	
	The proof of this theorem uses some techniques similar to the ones introduced for the classical Theorem \ref{monotoneXY}. These were indeed introduced by Ginibre \cite{Gin} in a general framework.
	As for the classical case, it is convenient to duplicate the model. We introduce a new doubled Hilbert space $\bar{\mathcal{H}}_\Lambda = \mathcal{H}_\Lambda\otimes\mathcal{H}_\Lambda$. Given an operator $\mathcal{O}$ acting on $\mathcal{H}_\Lambda$ we define two operators acting on $\bar{\mathcal{H}}_\Lambda$,
	\be
	\mathcal{O}_\pm = \mathcal{O}\otimes\mathds{1}\pm\mathds{1}\otimes\mathcal{O}.
	\ee
	The hamiltonian we consider for the doubled system is $H_{\Lambda,+}^\qu$:
	\be
	H_{\Lambda,+}^\qu =H_{\Lambda}^\qu\otimes\mathds{1}_\Lambda + \mathds{1}_\Lambda\otimes H_{\Lambda}^\qu = -\sum_{A\subset\Lambda}J^1_A (S^1_A)_+ + J^3_A(S^3_A)_+
	\label{H+}
	\ee
	The Gibbs state is denoted as
	\be
	\langle\langle O\rangle\rangle = \frac{1}{(Z^\qu_{\Lambda})^2} \Tr O\e{-H^\qu_{\Lambda,+}},
	\ee 
	for any operator $O$  acting on $\bar{\mathcal{H}}_\Lambda$.
	It follows from some straightforward algebra that
	\begin{eqnarray}
	\langle\mathcal{O}\mathcal{P}\rangle^\qu_{\Lambda}-\langle\mathcal{O}\rangle^{\qu}_{\Lambda}\langle\mathcal{P}\rangle^\qu_{\Lambda}&=&\frac{1}{2}\langle\langle \mathcal{O}_- \mathcal{P}_-\rangle\rangle;\label{doublecorr}\\
	(\mathcal{O}\mathcal{P})_\pm &=& \frac{1}{2}\left(\mathcal{O}_+\mathcal{P}_\pm + \mathcal{O}_-\mathcal{P}_\mp\right)\label{oppm},
	\end{eqnarray}
	for any $\mathcal{O}$, $\mathcal{P}$ operators on $\mathcal{H}_\Lambda$.
	
	Just as $\mathbb{C}^2$ constitutes the ``building block" for $\mathcal{H}_\Lambda$, so $\mathbb{C}^2 \otimes\mathbb{C}^2$  is to $\bar{\mathcal{H}}_\Lambda$. We can provide an explicit basis of $\mathbb{C}^2\otimes\mathbb{C}^2$ such that $S^1_+$, $S^1_-$, $S^3_+$, $-S^3_-$ have all non negative elements: 
	\begin{eqnarray}
	|\mu_+\rangle &=& \frac{1}{\sqrt{2}}\left(|++\rangle + |--\rangle\right),\;\;\;\;\;|\mu_-\rangle = \frac{1}{\sqrt{2}}\left(|++\rangle - |--\rangle\right),\\
	|\nu_+\rangle &=& \frac{1}{\sqrt{2}}\left(|+-\rangle + |-+\rangle\right),\;\;\;\;\;|\nu_-\rangle = \frac{1}{\sqrt{2}}\left(|+-\rangle - |-+\rangle\right).
	\end{eqnarray}
	Above by $|+\rangle$ and $|-\rangle$ we denote the basis of $\mathbb{C}^2$ formed by eigenvectors of $S^3$ with eigenvalues $\frac{1}{2}$ and $-\frac{1}{2}$ respectively, and $|i,j\rangle = |i\rangle\otimes|j\rangle$. It can be easily checked that the basis above has the required property. This result implies straightforwardly that there exists a basis of $\bar{\mathcal{H}}_\Lambda$ such that $(S^1_x)_+$, $(S^1_x)_-$, $(S^3_x)_+$ and $(-S^3_x)_-$ have non negative element for all $x\in\Lambda$.  
	Let us consider the truncated correlation function we are interested in:
	\be
	\Bigl\langle \prod_{x \in X} S_x^1 \prod_{x \in Y} S_x^1 \Bigr\rangle^{\qu}_{\Lambda} - \Bigl\langle \prod_{x \in X} S_x^1 \Bigr\rangle^{\qu}_{\Lambda} \; \Bigl\langle \prod_{x \in Y} S_x^1 \Bigr\rangle^{\qu}_{\Lambda}  = \tfrac{1}{2}\langle\langle\left(S^1_X\right)_- \left(S^1_Y\right)_-\rangle\rangle.
	\ee
	We can evaluate the right hand side of the equation above by a Taylor expansion:
	\be
	(Z^\qu_{\Lambda})^2 \langle\langle\left(S^1_X\right)_- \left(S^1_Y\right)_-\rangle\rangle = \sum_{n\geq 0}\frac{1}{n!}\Tr\left(S^1_X\right)_- \left(S^1_Y\right)_- (-H^\qu_{\Lambda,+})^n
	\ee
	Given the formulation of $H^\qu_{\Lambda,+}$ as in Eq. \ \eqref{H+} and relation \eqref{oppm}, it is clear that it can be expressed as a polynomial with positive coefficients of operators with nonnegative elements. The same holds for $(S^1_X)_-$ and$(S^1_Y)_-$. The trace of operators with nonnegative elements is non negative, hence the first inequality of the theorem. The second inequality can be proved precisely in the same way (with $S^2_Y$ substituted by $S^3_Y$), by noticing that $(S^3_Y)_-$ has necessarily non positive elements.
\end{proof}

Let us now turn to Theorem \ref{XYvsIsing}. While in the classical case it is necessary to introduce an artificial framework, interestingly the proof for the quantum case does not require such a construction. For spin-$\frac{1}{2}$ the statement can be easily recovered by recalling that the \emph{classical} Ising model can be recovered as a particular case of the \emph{ quantum} XY model (not of the classical one!). We review here a more general proof valid for any $\mathcal{S}$ \cite{Pe}.

\begin{proof}[Proof of Theorem \ref{XYvsIsing}]
	We reformulate the quantum Hamiltonian in order to have the interaction along the first and the third axis , as in Eq. \ \eqref{quantumXY13}. We prove here the following result, which is unitary equivalent to the statement of the theorem:
	\be
	\left\langle S^3_X\right\rangle^{{\rm qu}}_{\Lambda,\beta} \leq \mathcal{S}^{|X|}\langle s_A \rangle^{{\rm Is}}_{{\Lambda,\{\mathcal{S}^{|A|}J^3_A\}},\beta}.\label{statement2}
	\ee
	From now on we set $\beta = 1$ and drop all the dependencies on $\beta$ since it does not play any role. Let $\mathscr{S}^i_x = \mathcal{S}^{-1} S^i_x$ be the rescaled spin operators. The models we compare are the following:
	\begin{align}
	H^{\qu}_{\Lambda,\mathscr{S}} &= -\sum_{A\subset\Lambda}J^1_A \mathscr{S}^1_A + J^3_A \mathscr{S}^3_A,\\
	H^{{\rm{I}}}_{\Lambda,\{J^3_A\}} &= -\sum_{A\subset\Lambda}J^3_A s_A.
	\end{align}
	Clearly, \eqref{statement2} is equivalent to
	\be
	\langle \scrS^3_X\rangle^{\qu}_{{\Lambda,\mathscr{S}}}\leq \langle s_X\rangle^{{\rm{Is}}}_{{\Lambda,\{J^3_A\}}}.\label{statement3}
	\ee
	This is what we aim to prove. Let us now build a composite system where each site of the lattice hosts a quantum degree of freedom and an Ising variable at the same time. Let $\mathcal{H} = H^{\qu}_{\Lambda,\mathscr{S}} + H^{{\rm{I}}}_{\Lambda,\{J^3_A\}}$, i.e.
	\be
	\mathcal{H} = -\sum_{A\subset\Lambda}J^1_A \scrS^1_A + J^3_A(s_A +\scrS^3_A).
	\ee
	The Gibbs state is the natural one given the Gibbs states for the two separated systems. We denote it by $\langle \cdot\rangle_\Lambda$. We are interested in the expectation value $\langle s_X - \scrS^3_X\rangle_\Lambda$ for some $X\subset\Lambda$.  Since the trace is invariant under unitary transformations, we can apply on each site the unitary $(\scrS^1_x,\scrS^2_x,\scrS^3_x)\rightarrow (\scrS^1_x,s_x\scrS^2_s, s_x\scrS^3_x)$ and find
	\be
	\sum_{s\in\Omega_\Lambda}\Tr \left(s_X - \scrS^3_X\right)\e{-\mathcal{H}} = \sum_{s\in\Omega_\Lambda}\Tr s_X(1-\scrS^3_X)\e{\sum_{A\subset\Lambda}J^1_A \scrS^1_A +J^3_A s_A (1+\scrS^3_A)}
	\ee
	The expression evaluated above is just the expectation value we are interested in multiplied by the partition function of the system - which is positive and therefore not useful in the evaluation of the sign of $\langle s_X - \scrS^3_X\rangle_\Lambda$. By a Taylor expansion and by the property $\sum_{s\in\Omega_\Lambda} \prod_{x\in A}s_x^{n_x}\geq 0$ with $n_x\in\mathbb{N}$ for all $x\in\Lambda$ and any $A\subset\Lambda$, it is clear that the expression above is nonnegative. This implies that
	\be
	\langle s_X\rangle^{{\rm{Is}}}_{{\Lambda,\{J^3_A\}}}-\langle \scrS^3_X\rangle^{\qu}_{{\Lambda,\mathscr{S}}} = \langle s_X - \scrS^3_X\rangle_\Lambda\geq 0.
	\ee
	This proves Eq.\ \eqref{statement3}.
\end{proof}

\medskip
\noindent
{\bf Acknowledgements}:
The authors thank S. Bachmann and C.E. Pfister for useful comments.


\begin{thebibliography}{99.}
	
	\bibitem{Bach}
	S. Bachmann,
	\textit{Local disorder, topological ground state degeneracy and entanglement entropy, and discrete anyons},
	arXiv:1608.03903 (2016)
	
	\bibitem{BLU}
	C. Benassi, B. Lees, D. Ueltschi,
	\textit{Correlation inequalities for the quantum XY model}, J. Stat. Phys. 164, 1157--1166 (2016)
	
	\bibitem{BK}
	M.~Biskup, R.~Koteck\'{y},
	{\it True nature of long-range order in a plaquette orbital model},
	J. Statist. Mech. 2010, P11001 (2010)
	
	\bibitem{Dun}
	F. Dunlop, \textit{Correlation inequalities for multicomponent rotors}, Commun. Math. Phys. 49, 247--256 (1976)
	
	\bibitem{DLS}
	F.J.~ Dyson, E.H.~Lieb, B.~Simon,
	{\em Phase transitions in quantum spin systems with isotropic and nonisotropic interactions},
	J. Statist. Phys. 18, 335--383 (1978)
	
	\bibitem{FV}
	S. Friedli, Y. Velenik,
	\textit{Statistical Mechanics of Lattice Systems: a Concrete Mathematical Introduction}, \url{http://www.unige.ch/math/folks/velenik/smbook/index.html}
	
	\bibitem{FKG}
	C.M. Fortuin, P.W. Kasteleyn, J. Ginibre,
	\textit{Correlation inequalities on some partially ordered sets}, Commun. Math. Phys. 22, 89--103 (1971)
	
	\bibitem{FP}
	J. Fr{\"o}hlich, C.E. Pfister,
	\textit{Spin waves, vortices, and the structure of equilibrium states in classical XY  model}, Commun. Math. Phys. 89, 303--327 (1983) 
	
	\bibitem{FSS}
	J.~Fr\"ohlich, B.~Simon, T.~Spencer,
	{\em Infrared bounds, phase transitions and continuous symmetry breaking},
	Comm. Math. Phys. 50, 79--95 (1976)
	
	\bibitem{Gal}
	G. Gallavotti,
	\textit{A proof of the Griffiths inequalities for the X-Y model}, Stud. Appl. Math. 50, 89–92 (1971)
	
	\bibitem{Gin}
	J. Ginibre,
	\textit{General formulation of Griffiths' inequalities}, Commun. Math. Phys. 16, 310--328 (1970)
	
	\bibitem{Gri}
	R.B.~Griffiths,
	\textit{Correlations in Ising ferromagnets. I},
	J. Math. Phys. 8, 478--483 (1967)
	
	\bibitem{HM}
	M. Hasenbusch and S. Meyer,
	\textit{Critical exponents of the 3D XY model from cluster update Monte Carlo}, Phys. Lett. B 241.2, 238-242 (1990)
	
	\bibitem{HS}
	C.A. Hurst, S.Sherman, \textit{Griffiths' theorems for the ferromagnetic Heisenberg model}, Phys. Rev. Lett. 22, 1357 (1969)
	
	\bibitem{Kit}
	A.Y. Kitaev,
	\textit{Fault-tolerant quantum computation by anyons},
	Ann. Phys. 303:1 2--30 (2003)
	
	\bibitem{KPV75}
	H. Kunz, C.E. Pfister, P.A. Vuillermot,
	\textit{Correlation  inequalities for some classical spin vector models}, Phys. Lett. 54A, 428--430 (1975)
	
	\bibitem{KPV76}
	H. Kunz, C.E. Pfister, P.A. Vuillermot, 
	\textit{Inequalities for some classical spin vector models}, J. Phys. A: Math.Gen., Vol.9, No.10,  (1976) 
	
	\bibitem{MMP}
	A. Messager, S. Miracle-Sole, C. Pfister,
	\textit{Correlation inequalities and uniqueness of the equilbrium state for the plane rotator ferromagnetic model}, Commun. Math. Phys. 58, 19--29 (1978) 
	
	\bibitem{Mon}
	J.L. Monroe, \textit{Correlation inequalities for two-dimensional vector spin systems}, J. Math. Phys. 16, 1809--12 (1975)
	
	\bibitem{MP}
	J.L. Monroe, P.A. Pearce, \textit{Correlation inequalities for vector spin Models}, J. Stat. Phys, 21:615 (1979)
	
	\bibitem{Pe}
	P.A. Pearce, \textit{An inequality for spin-s X-Y ferromagnets}, Physics Letters A, 70 (2), 117-118 (1979)
	
	\bibitem{PM2}
	P.A. Pearce, J.L. Monoroe \textit{A simple proof of spin-$\frac{1}{2}$ X-Y inequalities}, Journal of Physics A: Mathematical and General, 12(7), L175 (1979)
	
	\bibitem{TVPS}
	T. Preis, P. Virnau, W. Paul and J.J. Schneider,
	\textit{GPU accelerated Monte Carlo simulation of the 2D and 3D Ising model}, J. Comp. Phys. 228(12) 4468--4477 (2009)
	
	\bibitem{Pr}
	C.J. Preston,
	\textit{A generalization of the FKG inequalities}, Comm. Math. Phys. 36, 233--241 (1974)
	
	\bibitem{Suz}
	M. Suzuki, \textit{Correlation inequalities and phase transition in the generalised X-Y model}, J. Math. Phys., 14, 837-838 (1973)
	
	\bibitem{U}
	D. Ueltschi,
	\textit{Random loop representations for quantum spin systems}, J. Math. Phys. 54(8), 083301 (2013)
	
	\bibitem{Vig}
	J.O. Vigfusson, \textit{Upper bound on the critical temperature in the 3D Ising model}, Journal of Physics A: Mathematical and General, 18(17) 3417 (1985)
	
	\bibitem{WJ}
	S.~Wenzel, W.~Janke,
	{\it Finite-temperature N\'eel ordering of fluctuations in a plaquette orbital model},
	Phys. Rev. B 80, 054403 (2009)
	
	\bibitem{Wes}
	S. Wessel, private communication.
	
	
	
	
	
	
	
	
\end{thebibliography}
\end{document}